\documentclass[11pt]{article}

\usepackage[margin=1in]{geometry}
\usepackage{amsmath,amsfonts,amssymb,amsthm}
\usepackage{graphicx}
\usepackage{enumerate}
\usepackage{bbm}
\usepackage{verbatim}
\usepackage{hyperref,color}
\usepackage[capitalize,nameinlink]{cleveref}
\usepackage[dvipsnames]{xcolor}
\hypersetup{
	colorlinks=true,
	pdfpagemode=UseNone,
    citecolor=OliveGreen,
    linkcolor=NavyBlue,
    urlcolor=Magenta,
	pdfstartview=FitW
}
\usepackage{appendix}
\crefname{appsec}{Appendix}{Appendices}
\usepackage{tikz}
\usepackage{subcaption}
\usetikzlibrary{calc}

\usepackage{afterpage}
\usepackage{array}
\usepackage{algorithm2e}
\SetKwComment{Comment}{/* }{ */}
\usepackage{multicol}

\usepackage{newtxtext}

\theoremstyle{plain}
\newtheorem{thm}{Theorem}[section]
\newtheorem{theorem}[thm]{Theorem}

\newtheorem{lemma}[thm]{Lemma}

\theoremstyle{definition}

\newtheorem*{ass*}{Assumption}
\newtheorem*{assumption*}{Assumption}

\theoremstyle{remark}

\newtheorem{remark}[thm]{Remark}

\crefname{lem}{Lemma}{Lemmas}
\crefname{lemma}{Lemma}{Lemmas}
\crefname{thm}{Theorem}{Theorems}
\crefname{theorem}{Theorem}{Theorems}
\crefname{defn}{Definition}{Definitions}
\crefname{definition}{Definition}{Definitions}
\crefname{fact}{Fact}{Facts}
\crefname{fact}{Fact}{Facts}
\crefname{clm}{Claim}{Claims}
\crefname{claim}{Claim}{Claims}
\crefname{prop}{Proposition}{Propositions}
\crefname{proposition}{Proposition}{Propositions}
\crefname{algocf}{Algorithm}{Algorithms}

\newcommand{\E}{\mathbb{E}}

\renewcommand{\nl}{{n,\lambda}}

\begin{document}

\title{On weighted partial triangulations of convex polygons}

 \author{
 	Antonio Blanca\thanks{Department of Computer Science and Engineering, Pennsylvania State University. Email: ablanca@cse.psu.edu. Research supported in part by NSF CAREER grant 2143762.}
 	\and
 	Alexandre Stauffer\thanks{Department of Mathematics, King's College London. Email: a.stauffer@kcl.ac.uk.} 
 	\and 
 	Izabella Stuhl\thanks{Department of Mathematics, Pennsylvania State University. Email: ius68@psu.edu. } 	
 }
\date{\today}
            
\maketitle

\begin{abstract}
We study the problem of sampling weighted partial triangulations of a convex polygon with $n+2$ sides.
We consider the distribution $\pi_\nl$ under which each
partial triangulation $\sigma$ is assigned probability proportional to $\lambda^{|\sigma|}$, where $\lambda>0$ is a model parameter and $|\sigma| \in \{0,\dots,n-1\}$ denotes the number of diagonals in $\sigma$. 
This model belongs to a broad class of weighted geometric partition problems that include lattice triangulations and dyadic tilings,
and is closely related to 
several classical combinatorial structures, including 
the full triangulations of a convex polygon and the associated Catalan structures.
Our main result is a simple exact sampling algorithm for $\pi_\nl$ with expected running time $O\big((\min\{n,n\sqrt{\lambda}\}+1)\log n\big)$, which is optimal up to the logarithmic factor.
\end{abstract}

\vfill

\thispagestyle{empty}

\pagebreak

\setcounter{page}{1} 
\section{Introduction}

We consider the problem of sampling weighted partial triangulations of a convex polygon. 
Let $\Omega_{n,k}$ denote the set of partitions of a convex polygon with exactly $n+2$ sides into $k$ parts obtained by inserting exactly $k-1$ non-intersecting diagonals; see Figures~\ref{fig:octagon-full}-\ref{fig:octagon-partialb}.
Given a real parameter $\lambda >0$, let $\pi_\nl$ be the distribution defined over 
$\Omega_n = \bigcup_{k=1}^n \Omega_{n,k}$ where each partial triangulation $\sigma \in \Omega_n$ is assigned probability
$$
\pi_\nl(\sigma) = \frac{\lambda^{|\sigma|}}{Z_\nl},
$$
where $|\sigma|$ denotes the number of diagonals in $\sigma$ and $Z_\nl = \sum_{\tau \in \Omega_n} \lambda^{|\tau|}$ is the corresponding normalizing constant or partition function.

Partial triangulations, also known as polygon dissections, belong to a broad class of geometric partition models which have a long history in combinatorics, computational geometry, and computer graphics. Classical triangulation models in this class include full and fixed-size triangulations of convex polygons, subdivisions and triangulations of arbitrary planar point sets, regular and Delaunay triangulations defined through various geometric lifting procedures, lattice triangulations, and triangulations of nonconvex polygonal regions or regions with holes. 
Other geometric partition models include mixed subdivisions, as well as dyadic, zonotopal, and cubical tilings; see the textbooks~\cite{lee2017subdivisions,de2010triangulations} for extensive background.

Across these models, sampling from uniform or weighted distributions 
over the corresponding partition spaces is a fundamental computational primitive
that is used both as a subroutine to find configurations that optimize a prescribed weight function and to investigate typical properties of random partitions. The weighted sampling problem has been studied, for example, in the context of lattice triangulations, dyadic tilings, and the full triangulations of convex polygons.
Lattice triangulations are triangulations of the integer points contained in a polygon in $\mathbb{R}^2$ whose vertices are also integer points; see Figure~\ref{fig:lattice}.
They are of interest as geometric structures in their own right and because of their connections to plane algebraic curves; see~\cite{ya1989real,dais2002resolving,gelfand1994discriminants,de2010triangulations}.
Motivated by questions about the typical geometry of random lattice triangulations,~\cite{caputo2013random,10.1214/16-EJP4321,stauffer2017lyapunov} studied Markov chains for the weighted sampling problem 
in which each triangulation $\tau$ is assigned probability $\propto \lambda^{|\tau|}$, with $|\tau|$ denoting the total edge length in $\tau$.

An analogous edge-length weighted sampling problem has been considered for dyadic tilings, which are rectangular dissections of an $n \times n$ lattice region into rectangles of area $n$ subject to dyadic restrictions on their dimensions and positions; see Figure~\ref{fig:dyadic} as well as~\cite{cannon2014phase,cannon2017polynomial,lagarias2002counting,angel2014phase}.
In this setting,~\cite{cannon2014phase} studied a Markov chain approach to sample from the distribution in which each dyadic tiling $\tau$ is assigned probability $\propto \lambda^{|\tau|}$, with $|\tau|$ corresponding to the total edge length in $\tau$.

Closer to the setting of this paper, uniform sampling has been studied extensively for the set of full triangulations of a convex polygon, which corresponds to $\Omega_{n,n}$ in our notation.
The number of triangulations in $\Omega_{n,n}$ is $C_n$, the $n$-th Catalan number, and 
$\Omega_{n,n}$ admits bijections with many classical Catalan structures including full binary trees with $n$ internal nodes, balanced parenthesizations, Dyck paths, among others~\cite{Stanley_2015}.
As such, sampling full triangulations uniformly at random is a relatively straightforward task, as one could, for example, use R{\'e}my's algorithm~\cite{remy1985procede} to sample a full binary tree with exactly $n$ internal vertices uniformly at random, and then utilize the bijection between such trees and the full triangulations~\cite{Stanley_2015}.

\renewcommand{\thesubfigure}{\alph{subfigure}}
\newcommand{\rect}[4]{\draw (#1,#2) rectangle (#3,#4);}
\afterpage{
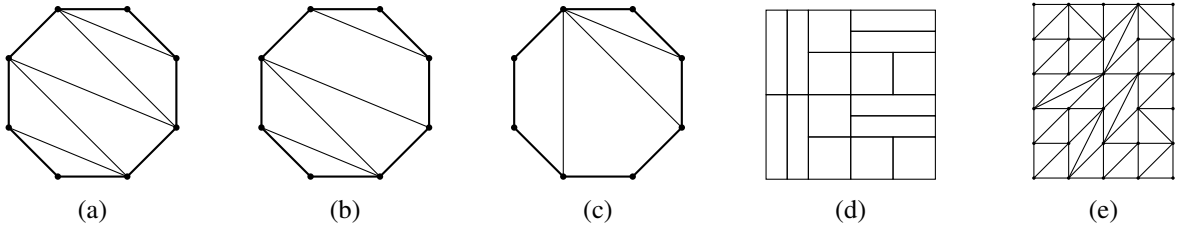
\begin{figure}[t]
  \centering

  \begin{subfigure}[b]{0.19\linewidth}
    \centering
    \begin{tikzpicture}[scale=1.2]
      \foreach \i in {0,1,...,7} {
        \coordinate (V\i) at ({360/8 * \i + 22.5}:1);
      }
      \draw[thick]
        (V0) -- (V1) -- (V2) -- (V3) -- (V4) -- (V5) -- (V6) -- (V7) -- cycle;
      \draw (V0) -- (V2);
      \draw (V7) -- (V2);
      \draw (V7) -- (V3);
      \draw (V6) -- (V3);
      \draw (V6) -- (V4);
      \foreach \i in {0,1,...,7} {
        \fill (V\i) circle (1pt);
      }
    \end{tikzpicture}
    \caption{}
    \label{fig:octagon-full}
  \end{subfigure}
  \hfill
  \begin{subfigure}[b]{0.19\linewidth}
    \centering
    \begin{tikzpicture}[scale=1.2]
      \foreach \i in {0,1,...,7} {
        \coordinate (V\i) at ({360/8 * \i + 22.5}:1);
      }
      \draw[thick]
        (V0) -- (V1) -- (V2) -- (V3) -- (V4) -- (V5) -- (V6) -- (V7) -- cycle;
      \draw (V0) -- (V2);
      \draw (V7) -- (V3);
      \draw (V6) -- (V3);
      \draw (V6) -- (V4);
      \foreach \i in {0,1,...,7} {
        \fill (V\i) circle (1pt);
      }
    \end{tikzpicture}
    \caption{}
    \label{fig:octagon-partial}
  \end{subfigure}
  \hfill
  \begin{subfigure}[b]{0.19\linewidth}
    \centering
    \begin{tikzpicture}[scale=1.2]
      \foreach \i in {0,1,...,7} {
        \coordinate (V\i) at ({360/8 * \i + 22.5}:1);
      }
      \draw[thick]
        (V0) -- (V1) -- (V2) -- (V3) -- (V4) -- (V5) -- (V6) -- (V7) -- cycle;
      \draw (V0) -- (V2);
      \draw (V7) -- (V2);
      \draw (V2) -- (V5);
      \foreach \i in {0,1,...,7} {
        \fill (V\i) circle (1pt);
      }
    \end{tikzpicture}
    \caption{}
    \label{fig:octagon-partialb}
  \end{subfigure}
  \hfill
  \begin{subfigure}[b]{0.19\linewidth}
    \centering
    \begin{tikzpicture}[scale=0.14]
      \rect{0}{8}{2}{16}
      \rect{0}{0}{2}{8}
      \rect{2}{8}{4}{16}
      \rect{2}{0}{4}{8}
      \rect{4}{12}{8}{16}
      \rect{4}{8}{8}{12}
      \rect{4}{4}{8}{8}
      \rect{4}{0}{8}{4}
      \rect{8}{14}{16}{16}
      \rect{8}{12}{16}{14}
      \rect{8}{8}{12}{12}
      \rect{12}{8}{16}{12}
      \rect{8}{6}{16}{8}
      \rect{8}{0}{12}{4}
      \rect{12}{0}{16}{4}
      \rect{8}{4}{16}{6}
    \end{tikzpicture}
    \caption{}
    \label{fig:dyadic}
  \end{subfigure}
  \hfill
  \begin{subfigure}[b]{0.19\linewidth}
    \centering
    \begin{tikzpicture}[scale=0.46]
      \draw (0,0) -- (0,1);  \draw (0,0) -- (1,0);  \draw (0,0) -- (1,1);
      \draw (0,1) -- (0,2);  \draw (0,1) -- (1,1);  \draw (0,1) -- (1,2);
      \draw (0,2) -- (0,3);  \draw (0,2) -- (1,2);  \draw (0,2) -- (1,3);
      \draw (0,2) -- (2,3);
      \draw (0,3) -- (0,4);  \draw (0,3) -- (1,3);  \draw (0,3) -- (1,4);
      \draw (0,4) -- (0,5);  \draw (0,4) -- (1,4);  \draw (0,4) -- (1,5);
      \draw (0,5) -- (1,5);
      \draw (1,0) -- (1,1);  \draw (1,0) -- (2,0);  \draw (1,0) -- (2,1);
      \draw (1,0) -- (2,2);
      \draw (1,1) -- (1,2);  \draw (1,1) -- (2,2);
      \draw (1,2) -- (2,2);  \draw (1,2) -- (2,3);
      \draw (1,3) -- (1,4);  \draw (1,3) -- (2,3);  \draw (1,3) -- (2,4);
      \draw (1,4) -- (1,5);  \draw (1,4) -- (2,4);
      \draw (1,5) -- (2,4);  \draw (1,5) -- (2,5);
      \draw (2,0) -- (2,1);  \draw (2,0) -- (3,0);  \draw (2,0) -- (3,1);
      \draw (2,1) -- (2,2);  \draw (2,1) -- (3,1);  \draw (2,1) -- (3,2);
      \draw (2,1) -- (3,3);
      \draw (2,2) -- (2,3);  \draw (2,2) -- (3,3);
      \draw (2,3) -- (2,4);  \draw (2,3) -- (3,3);  \draw (2,3) -- (3,4);
      \draw (2,3) -- (3,5);
      \draw (2,4) -- (2,5);  \draw (2,4) -- (3,5);
      \draw (2,5) -- (3,5);
      \draw (3,0) -- (3,1);  \draw (3,0) -- (4,0);  \draw (3,0) -- (4,1);
      \draw (3,1) -- (3,2);  \draw (3,1) -- (4,1);
      \draw (3,2) -- (3,3);  \draw (3,2) -- (4,1);  \draw (3,2) -- (4,2);
      \draw (3,2) -- (4,3);
      \draw (3,3) -- (3,4);  \draw (3,3) -- (4,3);  \draw (3,3) -- (4,4);
      \draw (3,4) -- (3,5);  \draw (3,4) -- (4,4);
      \draw (3,5) -- (4,4);  \draw (3,5) -- (4,5);
      \draw (4,0) -- (4,1);  \draw (4,1) -- (4,2);  \draw (4,2) -- (4,3);
      \draw (4,3) -- (4,4);  \draw (4,4) -- (4,5);
      \foreach \x in {0,1,2,3,4} {
        \foreach \y in {0,1,2,3,4,5} {
          \fill (\x,\y) circle (1.4pt);
        }
      }
    \end{tikzpicture}
    \caption{}
    \label{fig:lattice}
  \end{subfigure}

  \caption{Examples of geometric partitions.
    (\subref{fig:octagon-full}) A full triangulation of a regular octagon.
    (\subref{fig:octagon-partial})-(\subref{fig:octagon-partialb}) Two partial triangulations of the same octagon.
    (\subref{fig:dyadic}) A dyadic tiling of the unit square of size $n = 16$.
    (\subref{fig:lattice}) A lattice triangulation of a $4 \times 5$ rectangle.}
  \label{fig:subdivisions}
\end{figure}}

Markov chains on the space of full triangulations have also been studied extensively, and their convergence has proved to be quite difficult to analyze.
The natural diagonal-flip Markov chain, for instance, was conjectured to have
relaxation time $\Theta(n^{3/2})$~\cite{aldous1994triangulating}. Classical results established an $O(n^5 \log n)$ mixing time upper bound~\cite{mcshine1997mixing}
and an $\Omega(n^{3/2})$ lower bound~\cite{molloy1999mixing}.
The upper bound was recently improved to $O(n^{3}\log^3 n)$ in~\cite{eppstein2023improved}
and to $O(n^2 \mathrm{polylog}(n))$ in~\cite{alev2026faster}.
Despite these advances and substantial effort, the conjecture from~\cite{aldous1994triangulating} remains unresolved.

Partial triangulations have been studied in combinatorics since the 1850s and in computational geometry since the 1980s; see, e.g.,~\cite{cayley1890partitions,kirkman1857k,lee1989associahedron,stanley1996polygon,przytycki1998polygon,de2010triangulations,Stufler}. 
One motivation for this model is that  
as $\lambda\to\infty$, $\pi_\nl$ converges to the uniform distribution over full triangulations.

The uniform sampling problem for partial triangulations was studied in~\cite{panagiotou2023exact}, who gave an exact uniform sampler with $O(n)$ expected running time.
In the weighted setting, 
this problem raises a different type of algorithmic challenge. Namely, partial triangulations can be very sparse, so one would hope for an exact sampling algorithm
with running time that depends on the expected number of diagonals rather than on the polygon size $n$. As we shall see, this expectation is $O(n\sqrt{\lambda})$ and thus can be much smaller than $n$ when $\lambda$ is small. 

Direct algorithmic approaches, based either on explicitly computing $Z_\nl$ or on generating a full triangulation first, take time at least linear in $n$, independent of the number of diagonals in the resulting partial triangulation. Moreover, Markov chains on the space of partial triangulations with stationary distribution $\pi_\nl$ would produce only approximate (instead of exact) samples and, as the full triangulation case suggests, 
would likely mix in time polynomial in the typical number of diagonals, not to mention that 
establishing such a bound would likely inherit most of the same analytical difficulties encountered in the study of the diagonal-flip chain for full triangulations.
The framework from ~\cite{panagiotou2023exact} extends to $\pi_{n,\lambda}$, but the resulting sampler explicitly constructs a tree with $O(n)$ leaves and therefore requires at least linear running time.

Our main result is a simple, nearly-optimal exact sampling algorithm for $\pi_\nl$ that works for all $\lambda > 0$.

\begin{theorem}
\label{thm:intro}
There is a randomized algorithm that, 
for all $\lambda > 0$ and $n \ge 2$, 
outputs a partial triangulation distributed according to $\pi_\nl$ with expected running time $O\big((\min\{n,n\sqrt{\lambda}\}+1)\log n\big)$.
\end{theorem}

We prove that the expected number of diagonals under $\pi_{\nl}$ satisfies $\E_{\pi_\nl}[|\sigma|] = \Theta(\min\big\{n,n\sqrt{\lambda}\big\})$ when $\lambda > 1/n^2$, which implies that our result is optimal up to a logarithmic factor (see Lemma~\ref{lemma:expectation:bound} for more precise bounds on $\E_{\pi_\nl}[|\sigma|]$). 
Our algorithm works in two stages. First, it samples the number $X$ of diagonals in a partial triangulation drawn from $\pi_\nl$, and then generates a partial triangulation uniformly at random among all of those with exactly $X$ diagonals.
The algorithm to exactly sample the number of diagonals relies on the fact that we can efficiently find real numbers $Z_\nl^-$ and $Z_\nl^+$ such that $Z_\nl^- \le Z_\nl \le Z_\nl^+$ and 
$$
\frac{Z_\nl^+}{Z_\nl^-} \le 1+\frac{1}{3n^2}.
$$
We use this bound to design a natural Las Vegas procedure to generate a perfect sample from the target distribution over the number of diagonals.
At a high level, this procedure uses the upper and lower bounds on $Z_\nl$ to construct a simpler distribution and then corrects the bias via a rejection step. The bounds on $Z_\nl$ are obtained by noting that the ratio of the total contribution to $Z_\nl$ of configurations with $k+1$ vs. $k$ diagonals is decreasing in $k$ and decays at a geometric rate when $k 
\gtrsim n\sqrt{\lambda}$.

After generating the number of diagonals $X$ from the correct distribution, we identify a bijection between the partial triangulations in $\Omega_n$ and strings over the alphabet ')', '(', and '0' that are balanced in terms of open and close parentheses, start with '(' and have all zeros placed between pairs of open and close parentheses. 
To sample uniformly such a string,
our algorithm first uses R{\'e}my's algorithm~\cite{remy1985procede} for generating a full binary tree with $X+1$ internal vertices uniformly at random and translates the binary tree into a balanced parenthesization using a well-known Catalan bijection.
The resulting parenthesization is uniform over all balanced parenthesizations of that size, but not all of them have the same number of admissible positions where zeros can be inserted. Therefore, to insert the correct bias, we use another rejection step so that each parenthesization is accepted with probability proportional to the number of zero arrangements it permits.  Finally, the algorithm inserts the zeros by sampling a random partition consistent with the required placement constraints; for this we use an algorithm of Floyd for sampling subsets uniformly at random efficiently~\cite{bentley1987programming}.

\section{Sampling Algorithm}
\label{sec:alg}

Recall that $\Omega_{n,k}$ denotes the set of partitions of a convex polygon with exactly $n+2$ sides into $k$ parts. Hence, $\Omega_{n,n}$ is the set of full triangulations of the $(n+2)$-gon. We let $\Omega_n = \bigcup_{k=1}^n \Omega_{n,k}$ be the set of partial triangulations of the $(n+2)$-gon. 
We assume the RAM model of computation 
where arithmetic operations are assumed to take constant time, and we further assume access to a stream of perfectly random real numbers in $[0,1]$. These assumptions fix the computational model for consistency, but our algorithm extends to other models with minor modifications.

\subsection{Sampling the number of diagonals}

Our algorithm works in two stages: first we sample the number of diagonals and then a partial triangulation with exactly that number of diagonals uniformly at random.
We use the following facts in the first stage of the algorithm. 
The first one is a classical formula for the cardinality of $\Omega_{n,k}$,
stated by Kirkman~\cite{kirkman1857k} and first
proved by Cayley~\cite{cayley1890partitions}.
A modern proof of a more general formula is given in~\cite[Corollary~2]{przytycki1998polygon}.

\begin{lemma}\label{lemma:countomega}
For any $1\leq k\leq n$,
\begin{align}
   |\Omega_{n,k}| 
   = \frac{1}{n+1}\binom{n+k}{k}\binom{n-1}{k-1}. 
   \label{eq:countparttriang}
\end{align}
\end{lemma}

\begin{lemma}
\label{lemma:new:ratio:bound}
For every $n\geq 2$ and $\lambda>0$, there exist real numbers
$Z_\nl^-,Z_\nl^+>0$ such that
$ Z_\nl^- \leq Z_\nl \leq Z_\nl^+$ and
    $$
    \frac{Z_\nl^+}{Z_\nl^-}
    \leq 1+\frac{1}{3n^2}.
    $$
Moreover, $Z_\nl^-$ and $Z_\nl^+$ can be computed in
$O\big(\E_{\pi_{n,\lambda}}[|\sigma|]+\log n
\big)$ time.
\end{lemma}

We prove Lemma~\ref{lemma:new:ratio:bound} in Section~\ref{subsec:ratio}.
Consider now the distribution $\pi_\nl^\perp$ on $\{0,\dots,n-1\}$ where 
$$\pi_\nl^\perp(k) = \frac{|\Omega_{n,{k+1}}|\lambda^k}{Z_\nl}.$$ 
This distribution is the projection of $\pi_\nl$ to the number of diagonals.
To sample from $\pi_\nl$, we first design an algorithm to sample from $\pi_\nl^\perp$.

Using the upper bound $Z_\nl^+$ on $Z_\nl$ from  Lemma~\ref{lemma:new:ratio:bound}, we define for $k \in \{0,\dots,n-1\}$ the sequence:
$$
p_k = \frac{|\Omega_{n,{k+1}}|\lambda^k}{ Z_\nl^+ \prod_{j=0}^{k-1} (1-p_j)}.
$$
We assume that $\prod_{j=0}^{k-1} (1-p_j) = 1$ when $k=0$ so that
$p_0 = |\Omega_{n,1}|/Z_\nl^+=1/Z_\nl^+$. 

Our algorithm for sampling from $\pi_\nl^\perp$ proceeds by rounds.
In each round,
starting with $i = 0$, the algorithm outputs $i$ with probability $p_i$; otherwise, $i$ is increased by $1$ and the process is repeated. 
If the algorithm fails to output any $i \in \{0,\dots,n-1\}$,
a new round is started.
The steps of the algorithm are given in detail below.

\medskip
\RestyleAlgo{ruled}
\begin{algorithm}[H]
	\caption{Exact sampler for $\pi_\nl^\perp$}
	\label{alg:exact:k:sampler}
	
    $R = 1$\;    
	 \While{true}{
	 	\For{$i = 0, \dots, n-1$}{
	 		Sample $r \in [0,1]$ uniformly at random\; 

            Compute $p_{i}$;
            
	 		\If{$r \le p_i$}{\KwOut{$i$}}
	 	}
         $R = R+1$;
	 }
\end{algorithm}

\bigskip
We justify the correctness and analyze the running time of Algorithm \ref{alg:exact:k:sampler} next.
\begin{lemma}
\label{lemma:alg:exact:k:sampler:correctness}
The output of Algorithm \ref{alg:exact:k:sampler} is distributed according to $\pi_\nl^\perp$.
\end{lemma}
\begin{proof}
We first check that $p_k \in (0,1)$ for every $k \in \{0,\dots,n-2\}$ and $p_{n-1} \in (0,1]$. For this, we can show inductively
that for every $k \in \{0,\dots,n-1\}$
\begin{equation}
\label{eq:pk:product}
    \prod_{j=0}^{k}(1-p_j)
    =
    1-
    \frac{1}{Z_\nl^+}
    \sum_{j=0}^{k}|\Omega_{n,j+1}|\lambda^j.
\end{equation}
Indeed, for $k=0$, 
$1-p_0
    =
    1-{1}/{Z_\nl^+}
    =
    1-{|\Omega_{n,1}|}/{Z_\nl^+}
$ since $|\Omega_{n,1}|=1$. Moreover,
\begin{align*}
    \prod_{j=0}^{k}(1-p_j)
    &=
    \prod_{j=0}^{k-1}(1-p_j)(1-p_k)
    =
    \frac{
        Z_\nl^+
        -
        \sum_{j=0}^{k-1}|\Omega_{n,j+1}|\lambda^j
    }{Z_\nl^+}
    \left(
        1-
        \frac{|\Omega_{n,k+1}|\lambda^k}{
            Z_\nl^+
            -
            \sum_{j=0}^{k-1}|\Omega_{n,j+1}|\lambda^j
        }
    \right)\\
    &=
    1-
    \frac{1}{Z_\nl^+}
    \sum_{j=0}^{k}|\Omega_{n,j+1}|\lambda^j.
\end{align*}
Hence, we get
$
p_k =  \frac{|\Omega_{n,{k+1}}|\lambda^k}{ Z_\nl^+ -
    \sum_{j=0}^{k-1}|\Omega_{n,j+1}|\lambda^j}.
$
For $k \in \{0,\dots,n-1\}$, we observe that
\begin{align*}
Z_\nl^+
-\sum_{j=0}^{k-1}|\Omega_{n,j+1}|\lambda^j
-|\Omega_{n,k+1}|\lambda^k
=
Z_\nl^+
-\sum_{j=0}^{k}|\Omega_{n,j+1}|\lambda^j
\geq
Z_\nl
-\sum_{j=0}^{k}|\Omega_{n,j+1}|\lambda^j.
\end{align*}
This difference is strictly positive for $k \in \{0,\dots,n-2\}$ and nonnegative when $k=n-1$; this implies that $p_k \in (0,1)$ for $k \in \{0,\dots,n-2\}$ and $p_{n-1} \in (0,1]$.

For the output distribution of the algorithm, note that for any $\ell\ge 1$, conditional on the $\ell$-th round being reached, the probability that the algorithm outputs in that round is
$$
\sum_{k=0}^{n-1} p_k \prod_{j=0}^{k-1} (1-p_j) = \sum_{k=0}^{n-1} \frac{|\Omega_{n,k+1}|\lambda^k}{ Z_\nl^+} = \frac{ Z_\nl}{Z_\nl^+}.
$$
Hence, conditioned on the algorithm producing an output on the $\ell$-th round, the probability that the output is $k \in \{0,\dots,n-1\}$ is
\begin{equation}\label{eq:rej}
p_k \Big(\prod_{j=0}^{k-1} (1-p_j)\Big) \frac{ Z_\nl^+ }{ Z_\nl } = \pi_\nl^\perp(k),
\end{equation}
as claimed.
\end{proof}

\begin{lemma}
\label{lemma:alg1:rt}
The expected running time of Algorithm~\ref{alg:exact:k:sampler} is $O(\E_{\pi_\nl}[|\sigma|] +\log n)$.
\end{lemma}
\begin{proof}
     We use the upper bound $Z_\nl^+$ on $Z_\nl$ from Lemma~\ref{lemma:new:ratio:bound}; note that
    $Z_\nl^+$ can be computed in $O(\E_{\pi_\nl}[|\sigma|] +\log n)$ time, and thus so can
    $p_0 = 1/Z_\nl^+$. In addition, from Lemma~\ref{lemma:countomega} we obtain for $k \in \{0,\dots,n-2\}$ that    
    $$
    p_{k+1} = \frac{|\Omega_{n,k+2}| \lambda p_k}{|\Omega_{n,k+1}|(1-p_k)} = \frac{(n+k+2)(n-k-1)}{(k+1)(k+2)} \cdot \frac{\lambda p_k}{(1-p_k)};
    $$
    hence, each $p_k$ can be computed in $O(1)$ time from $p_{k-1}$, and thus the cost of computing $p_i$ in each round is $O(1)$.
    
    Let $T$ denote the random variable corresponding to the total running time of the ``while'' loop of Algorithm~\ref{alg:exact:k:sampler}, $I_\ell$  the number of iterations of the ``for'' loop when $R = \ell$, and $R_{\textsc{out}}$ the value of $R$ when the algorithm outputs.
    Then, 
 \begin{align*}
         \E[T] &= \sum_{\ell \ge 1} \E[T \mid R_{\textsc{out}} = \ell] \Pr[R_{\textsc{out}} 
         = \ell] \\&= O(1)\sum_{\ell \ge 1} (n(\ell-1) + \E[I_\ell\mid \text{round $\ell$ outputs}])\Pr[R_{\textsc{out}} = \ell],
    \end{align*}   
    and using~\eqref{eq:rej}, we obtain
    \begin{align*}
    \E[I_\ell \mid \text{round $\ell$ outputs}] &= \sum_{k=0}^{n-1} (k+1) \Pr[I_\ell = k+1\mid\text{round $\ell$ outputs}]
    \\&= 
    \sum_{k=0}^{n-1} (k+1) \frac{p_k \prod_{j=0}^{k-1} (1-p_j)}{Z_\nl/Z_\nl^+} = \sum_{k=0}^{n-1} \frac{(k+1) |\Omega_{n,k+1}|\lambda^k}{Z_\nl}  
    \\
    &= \E_{\pi_\nl}[|\sigma|]+1.    
    \end{align*}
    Therefore,
    \begin{align*}
    \E[T] &= O(\E_{\pi_\nl}[|\sigma|]+1) + O(n) \sum_{\ell \ge 1} (\ell-1)\Pr[R_{\textsc{out}}=\ell] \\
    &= O(\E_{\pi_\nl}[|\sigma|]+1) + O(n)(\E[R_{\textsc{out}}]-1).
    \end{align*}
    Finally, note that $R_{\textsc{out}}$ is a geometric random variable with support $\{1,2,\dots\}$. 
    Therefore,
    by Lemma~\ref{lemma:new:ratio:bound} 
    $$\E[R_{\textsc{out}}] = \frac{Z_\nl^+}{Z_\nl} \le \frac{Z_\nl^+}{Z_\nl^-} \le  1 + \frac{1}{3n^2},$$
    and we can conclude that $\E[T] = O(\E_{\pi_\nl}[|\sigma|] +1)$. So, the overall running time is dominated by the computation of $p_0$.
\end{proof}

\begin{remark}
We remark that there is a somewhat equivalent cumulative sum implementation of each round of Algorithm~\ref{alg:exact:k:sampler}.
In particular, one can sample $r$ uniformly from $[0,Z_\nl^+]$, compute the
terms $|\Omega_{n,k+1}|\lambda^k$ sequentially, and output the first $k$ such that $r \leq
    \sum_{j=0}^{k}|\Omega_{n,j+1}|\lambda^j$.
If no such $k$ exists, the round is rejected.
Under our computational model,
the two implementations are equivalent and have the same expected
running time up to constant factors.
\end{remark}

\subsection{Bounding the expected number of diagonals}

We provide next tight bounds on the expected number of diagonals under $\pi_{\nl}$.
This will be used in the proof of Lemma~\ref{lemma:new:ratio:bound} and 
to connect the $O(\E_{\pi_\nl}[|\sigma|] +\log n)$
running time of Algorithm~\ref{alg:exact:k:sampler} from Lemma~\ref{lemma:alg1:rt} to the claimed running time in Theorem~\ref{thm:intro} from the introduction.

\begin{lemma}
\label{lemma:expectation:bound}
For every $n\geq2$ and $\lambda>0$, the following bounds hold:
\begin{enumerate}
    \item If $\lambda>n^{-2}$, then
    \[
        \frac{1}{16}\min\big\{n,n\sqrt{\lambda}\big\}
        \le
        \E_{\pi_{n,\lambda}}[|\sigma|]
        \le
        4\min\big\{n,n\sqrt{\lambda}\big\}.
    \]
     \item If $\lambda \le n^{-2}$, then
    \[
        \frac{\lambda n^2}{4}
        \le
        \E_{\pi_{n,\lambda}}[|\sigma|]
        \le
        16\lambda n^2.
    \]
\end{enumerate}
\end{lemma}
\begin{proof}
For $k \in \{1,\dots,n\}$, let
$$
    W_k
    =
    \lambda^{k-1}|\Omega_{n,k}|
    =
    \frac{\lambda^{k-1}}{n+1}
    \binom{n+k}{k}\binom{n-1}{k-1},
$$
so that $Z_\nl=\sum_{k=1}^n W_k$;
for convenience, we set $W_{n+1} = 0$. Then, 
\begin{align}
    \frac{W_{k+1}}{W_k}
    &=
    \lambda
    \frac{\binom{n+k+1}{k+1}}{\binom{n+k}{k}}
    \frac{\binom{n-1}{k}}{\binom{n-1}{k-1}}
    =
    \lambda
    \frac{(n+k+1)(n-k)}{k(k+1)}
    =
    \lambda\left(
        \frac{n(n+1)}{k(k+1)}-1
    \right).
    \label{eq:weight:ratio}
\end{align}
From this, we see that $\frac{W_{k+1}}{W_k}$ is strictly decreasing in $k$ and that $\frac{W_{k+1}}{W_k} < 1$ when
$k(k+1) > \frac{\lambda}{1+\lambda}n(n+1)$. So, let $s = \sqrt{
            \frac{\lambda}{1+\lambda}n(n+1)
        }$ and let $b
    =
    \lceil
        2s
    \rceil.$
    We show first that $\E_{\pi_{n,\lambda}}[|\sigma|]<b$ for all $\lambda > 0$.
When $b\geq n$, we trivially have
$
    \E_{\pi_{n,\lambda}}[|\sigma|]
    \leq n-1
    <b.
$
Suppose then that $b<n$ and observe that $
    \lambda<\frac{n}{3n+4}<\frac{1}{3}$
and that for $k \ge b$ we have
\[
    k(k+1)\geq k^2 \geq 4
        \frac{\lambda}{1+\lambda}n(n+1).
\]
Combining these facts with~\eqref{eq:weight:ratio} we obtain
\begin{align}
    \frac{W_{k+1}}{W_k}
    &\leq
    \lambda\left(
        \frac{1+\lambda}{4\lambda}-1
    \right)
    =
    \frac{1-3\lambda}{4}
    \leq \frac{1}{4}.
    \label{eq:weight:geometric}
\end{align}
Consequently, for every $j\in\{1,\dots,n-b\}$, we have
\[
    W_{b+j}
    \leq
    \frac{1}{4^j}W_b.
\]
Using this fact:
\begin{align*}
     \E_{\pi_{n,\lambda}}[|\sigma|] = \sum_{k=1}^n(k-1) \frac{W_k}{Z_\nl}
    &=
    \frac{1}{Z_\nl} \sum_{k=1}^b(k-1)W_k
    +
      \frac{1}{Z_\nl} \sum_{j=1}^{n-b}(b+j-1)W_{b+j}\\
    &\leq
      \frac{b-1}{Z_\nl}\sum_{k=1}^bW_k
    +
    \frac{1}{Z_\nl} \sum_{j=1}^{n-b}\bigl(b-1+j\bigr)W_{b+j}\\
    &=
    b-1
    +
     \frac{1}{Z_\nl} \sum_{j=1}^{n-b}jW_{b+j}\\
    &\leq
    b-1+
    \frac{W_b}{Z_\nl}
    \sum_{j=1}^{\infty}\frac{j}{4^j}\\
    &\leq
    b-1+\frac49
    <b.
\end{align*}
We show next that $\E_{\pi_{n,\lambda}}[|\sigma|]
    \ge b/16$ when $\lambda > 1/n^2$.
For this, let us assume first that $s \ge 2$. 
Then, for
$k \le  s/2$ we have
\[
    k(k+1)
    \leq
    \frac{s}{2}\left(\frac{s}{2}+1\right)
    \le
    \frac{s^2}{2},
\]
and it follows from~\eqref{eq:weight:ratio} that
\[
    \frac{W_{k+1}}{W_k}
    \ge
    \lambda\left(
        \frac{2(1+\lambda)}{\lambda}-1
    \right)
    =
    2+\lambda
    >
    2.
\]
Also, $s/2<n$ since $s^2<n(n+1)<4n^2$. Hence
$\lfloor s/2\rfloor+1\le n$, and
\[
    \sum_{k=1}^{\lfloor s/2\rfloor}W_k
    \le W_{\lfloor s/2\rfloor}
    \sum_{j=0}^{\lfloor s/2\rfloor-1}\frac{1}{2^j}
    \le
    2W_{\lfloor s/2\rfloor}
    <
    W_{\lfloor s/2\rfloor+1}
    \leq
    \sum_{k=\lfloor s/2\rfloor+1}^{n}W_k.
\]
Consequently,
$
    \Pr_{\pi_{n,\lambda}}\left(
        |\sigma|\geq\left\lfloor s/2 \right\rfloor
    \right)
    \geq\frac12,$
and therefore when $s \ge 2$
\[
    \E_{\pi_{n,\lambda}}[|\sigma|]
    \ge
    \frac12\left\lfloor\frac s2\right\rfloor
    \ge
    \frac{b}{16}.
\]

By Markov's inequality, using the fact that $x/(1+x)$ is an increasing function, we have
\begin{align}
    \E_{\pi_{n,\lambda}}[|\sigma|]
    &\ge
    \Pr_{\pi_{n,\lambda}}(|\sigma| \geq1)
    =
    \frac{\sum_{k=2}^nW_k}
         {1+\sum_{k=2}^nW_k}
    \geq
    \frac{W_2}{1+W_2}. \label{eq:w2:lower}
\end{align}
When $s<2$, since $\lambda>1/n^{2}$, we also have
\[
    W_2
    =
    \frac{\lambda}{2}(n-1)(n+2)
    \geq
    \frac{\lambda n^2}{2}
    >
    \frac12,
\]
and thus $\E_{\pi_{n,\lambda}}[|\sigma|] > \frac 13 > \frac{b}{16}$.

Now, observe that $n < \sqrt{n(n+1)}<3n/2$, and 
when $\lambda>1/n^{2}$ we have
\[
    \frac{1}{2}\min\{1,\sqrt{\lambda}\}
    <
    \sqrt{\frac{\lambda}{1+\lambda}}
    <
    \min\{1,\sqrt{\lambda}\}.
\]
Consequently,
$
    b
    \geq
    2s
    >
    \min\{n,n\sqrt{\lambda}\},
$
and
\[
    b
    <
    2s+1
    <
    3n\min\{1,\sqrt{\lambda}\}+1
    <
    4\min\{n,n\sqrt{\lambda}\}.
\]
It follows that
\[
    \frac{1}{16}\min\{n,n\sqrt{\lambda}\}
    <
    \E_{\pi_{n,\lambda}}[|\sigma|]
    <
    4\min\{n,n\sqrt{\lambda}\}.
\]

We consider now the $\lambda\leq 1/n^{2}$ case.
Since the ratio $W_{k+1}/W_k$ is
decreasing in $k$, we have that for $n \ge 2$
\begin{align}
\label{eq:w2:upper}
    \frac{W_{k+1}}{W_k}
    \le \frac{W_2}{W_1} =
    \lambda\left(
        \frac{n(n+1)}{2}-1
    \right) \le 
    \frac{n+1}{2n} \le \frac{3}{4},
\end{align}
and then
\begin{equation*}
    W_{j+1}
    \leq
    \left(\frac{3}{4}\right)^{j-1}W_2
\end{equation*}
for every $j\in\{1,\dots,n-1\}$. Since
$Z_{n,\lambda}\geq1$, it follows that
\begin{align*}
   \E_{\pi_{n,\lambda}}[|\sigma|]
    &=
    \frac{1}{Z_{n,\lambda}}
    \sum_{j=1}^{n-1}jW_{j+1}
    \le
    W_2
    \sum_{j=1}^{\infty}
    j\left(\frac{3}{4}\right)^{j-1}
=
    16W_2.
\end{align*}
From~\eqref{eq:w2:upper}, we know that $W_2\leq3/4$, so from~\eqref{eq:w2:lower} we obtain
$
    \E_{\pi_{n,\lambda}}[|\sigma|]
    \ge
    \frac{W_2}{2}.
$
Also, for $n\geq2$,
$
    \frac{\lambda n^2}{2}
    \leq
    W_2
    \leq
    {\lambda n^2}
$ and thus we conclude that $
      \frac{\lambda n^2}{4} \le \E_{\pi_{n,\lambda}}[|\sigma|] \le {16 \lambda n^2},
$ as claimed.
\end{proof}

\subsection{Bounding the partition function ratio: proof of Lemma~\ref{lemma:new:ratio:bound}}
\label{subsec:ratio}

We prove next the bound on the ratio ${Z_\nl^+}/{Z_\nl^-}$
in Lemma~\ref{lemma:new:ratio:bound} which is key in justifying the efficiency of Algorithm~\ref{alg:exact:k:sampler}. 

\begin{proof}[Proof of Lemma~\ref{lemma:new:ratio:bound}]
As in the proof of Lemma~\ref{lemma:expectation:bound}, we let 
$$
    W_k
    =
    \lambda^{k-1}|\Omega_{n,k}|
    =
    \frac{\lambda^{k-1}}{n+1}
    \binom{n+k}{k}\binom{n-1}{k-1}
$$
for $k \in \{1,\dots,n\}$ and set $W_{n+1} = 0$.
Recall from~\eqref{eq:weight:ratio} that the ratio $\frac{W_{k+1}}{W_k}$ is strictly decreasing in $k$ and that $\frac{W_{k+1}}{W_k} < 1$ when
$k(k+1) > \frac{\lambda}{1+\lambda}n(n+1)$. Therefore, 
let $b = \Big\lceil 2\sqrt{
        \frac{\lambda}{1+\lambda}n(n+1)
    }\Big\rceil$ and define the threshold
$$
m = \min\Big\{n,b+\big\lceil2\log_4 n\big\rceil\Big\}.
$$
Set
$Z_\nl^-
    =
    \sum_{k=1}^{m}W_k$ and $
    Z_\nl^+
    =
    Z_\nl^-+\frac{4}{3}W_{m+1}$.
    
If $m=n$, then
$W_{m+1}=W_{n+1}=0,$ and
$
    Z_\nl^-=Z_\nl=Z_\nl^+.$ 
Thus, the required bound on the ratio holds. 
Let us assume then that $m<n$.
In this case,
$
   b < n
$
which implies that 
\begin{equation}
    \label{eq:l:b}
    \lambda<\frac{n}{3n+4}<\frac{1}{3}.
\end{equation}
Moreover, for $k \ge b$ we have
\[
    k(k+1)\geq k^2 \geq 4
        \frac{\lambda}{1+\lambda}n(n+1),
\]
and combining this inequality with~\eqref{eq:weight:ratio} and~\eqref{eq:l:b} we obtain that
\begin{align*}
    \frac{W_{k+1}}{W_k}
    &\leq
    \lambda\left(
        \frac{1+\lambda}{4\lambda}-1
    \right)
    =
    \frac{1-3\lambda}{4}
    \leq \frac{1}{4}.    
\end{align*}
Since $m\geq b$, it follows that
\begin{align*}
    Z_\nl-Z_\nl^-
    &=
    \sum_{k=m+1}^n W_k\leq
    W_{m+1}
    \sum_{j=0}^{\infty}\frac{1}{4^j}
    =
    \frac{4}{3}W_{m+1}
    =
    Z_\nl^+-Z_\nl^-,
\end{align*}
which establishes that $Z_\nl\leq Z_\nl^+$.

Now, to bound the ratio of $Z_\nl^+/Z_\nl^-$ observe that
$$
    W_{m+1}
    \leq
    \frac{1}{4^{\lceil 2\log_4n+1\rceil}}W_{b} \le \frac{W_b}{4n^2}.
$$
Then,
\begin{align*}
    \frac{Z_\nl^+}{Z_\nl^-}
    &=
    1+\frac 43 \cdot \frac{W_{m+1}}{Z_\nl^-}
    \leq
    1+\frac{4}{3} \frac{W_{b}}{4n^2Z_\nl^-}
    \leq
    1+\frac{1}{3n^2},
\end{align*}
where the last inequality follows from the fact that $W_{b}\le Z_\nl^-$ since $b \le m$.

Finally, note that we can compute 
$Z_\nl^-$ and $Z_\nl^+$ in $O(m)$ time 
since $W_{k+1}$ can be computed from $W_k$
using~\eqref{eq:weight:ratio}. Moreover,
\[
    m
    \leq
    b+\left\lceil2\log_4 n\right\rceil
    =
    O\big(\min\{n,n\sqrt{\lambda}\}+\log n\big).
\]
Then, if $\lambda>n^{-2}$, Lemma~\ref{lemma:expectation:bound} implies that
$\min\{n,n\sqrt{\lambda}\}=O(\E_{\pi_\nl}[|\sigma|])$, while if
$\lambda\leq n^{-2}$, then $n\sqrt{\lambda}\leq1$. Therefore,
$
    m=O\big(\E_{\pi_\nl}[|\sigma|]+\log n\big)
$
as required.
\end{proof}

\section{Sampling a partial triangulation with a fixed number of diagonals} 

In the previous section we discussed how to sample the number of diagonals in a partial triangulation from the correct distribution $\pi_\nl^\perp$.
Our next goal is to sample a partial triangulation with a given number of diagonals uniformly at random.

Our algorithm utilizes a bijection between the set of partial triangulations with exactly $k$ parts and strings
 of symbols '(', ')', and '0', where parentheses must 
 be balanced (i.e., respect the open-close structure) and any zero must be inside a pair of open-close parentheses. 
Let $\Upsilon_{m,j}$ be the set of such strings with $m$ pairs of open-close parentheses and $j$ zeros. We 
construct a bijection between the sets $\Omega_{n,k}$ and $\Upsilon_{k,n-k}$.

\begin{figure}[t]
    \centering
    \begin{subfigure}[b]{0.35\linewidth}
   \begin{tikzpicture}[scale=2, line width=0.6pt,baseline=(current bounding box.north)]

  \coordinate (V0) at (240:1);
  \coordinate (V1) at (300:1);
  \coordinate (V2) at (0:1);
  \coordinate (V3) at (60:1);
  \coordinate (V4) at (120:1);
  \coordinate (V5) at (180:1);

  \draw[thick] (V1) -- (V2) -- (V3) -- (V4) -- (V5) -- (V0);
  \draw[thick] (V0) -- (V1);
  \draw (V1) -- (V5);
  \draw (V2) -- (V5);

  \coordinate (F1) at ($(V0)!0.5!(V5)!0.45!(V1)$);
  \coordinate (F2) at ($(V1)!0.5!(V5)!0.45!(V2)$);
  \coordinate (F3) at (barycentric cs:V2=1,V3=1,V4=1,V5=1);

  \coordinate (X0) at ($(V0)!0.5!(V1)+(0,-0.3)$);
  \coordinate (X1) at ($(V5)!0.5!(V0)+(-0.3,-0.15)$);
  \coordinate (X2) at ($(V1)!0.5!(V2)+(0.3,-0.15)$);
  \coordinate (X3) at ($(V2)!0.5!(V3)+(0.3,0.15)$);
  \coordinate (X4) at ($(V3)!0.5!(V4)+(0,0.3)$);
  \coordinate (R)  at ($(V4)!0.5!(V5)+(-0.3,0.15)$);

  \draw[blue, thick] (X0) -- (F1);
  \draw[blue, thick] (X1) -- (F1);
  \draw[blue, thick] (X2) -- (F2);
  \draw[blue, thick] (R)  -- (F3);
  \draw[blue, thick] (X4) -- (F3);
  \draw[blue, thick] (X3) -- (F3);
  \draw[blue, thick] (F2) -- (F3);
  \draw[blue, thick] (F1) -- (F2);

  \node at ($(F3)!0.25!(R) +(-0.10, 0.09)$) {\small $e_{-1}$};
  \node at ($(F3)!0.25!(X4)+( 0.12, 0.06)$) {\small $e_0$};
  \node at ($(F3)!0.25!(X3)+( 0.2,-0.07)$) {\small $e_1$};
  \node at ($(F3)!0.25!(F2)+( 0.15, -0.05)$) {\small $e_2$};

  \fill[blue] (F1) circle (0.8pt);
  \fill[blue] (F2) circle (0.8pt);
  \fill[blue] (F3) circle (0.8pt);

  \fill[blue] (X0) circle (1pt);
  \fill[blue] (X1) circle (1pt);
  \fill[blue] (X2) circle (1pt);
  \fill[blue] (X3) circle (1pt);
  \fill[blue] (X4) circle (1pt);

  \fill[red] (R) circle (1.6pt);

  \foreach \i in {0,1,2,3,4,5} {
    \fill (V\i) circle (1.2pt);
  }

\end{tikzpicture}
\caption{}
\end{subfigure}
\begin{subfigure}[b]{0.35\linewidth}
    \begin{tikzpicture}[scale=1.8, line width=0.6pt,baseline=(current bounding box.north)]
 
  \foreach \i in {0,1,2,3,4,5,6,7} {
    \coordinate (V\i) at ({-67.5 + 45*\i}:1);
  }
    \draw[thick] (V0) -- (V1) -- (V2) -- (V3) -- (V4);
  \draw[red, ultra thick] (V4) -- (V5); 
  \draw[thick] (V5) -- (V6) -- (V7) -- (V0);
  
  \draw (V2) -- (V4);
  \draw (V2) -- (V6);

  \coordinate (F1) at (barycentric cs:V2=1,V4=1,V5=1,V6=1);
  \coordinate (F2) at (barycentric cs:V2=1,V3=1,V4=1);
  \coordinate (F3) at (barycentric cs:V0=1,V1=1,V2=1,V6=1,V7=1);

  \coordinate (L1) at ($(V3)!0.5!(V4)+(0, 0.30)$);   
  \coordinate (L2) at ($(V2)!0.5!(V3)+(0.27, 0.20)$);
  \coordinate (L3) at ($(V1)!0.5!(V2)+(0.30, 0)$);  
  \coordinate (L4) at ($(V0)!0.5!(V1)+(0.27,-0.20)$);
  \coordinate (L5) at ($(V7)!0.5!(V0)+(0,-0.30)$);   
  \coordinate (L6) at ($(V6)!0.5!(V7)+(-0.27,-0.20)$); 
  \coordinate (L7) at ($(V5)!0.5!(V6)+(-0.30, 0)$); 
  
  \coordinate (R) at ($(V4)!0.5!(V5)+(-0.27, 0.20)$);

  \draw[blue, thick] (R)  -- (F1);
  \draw[blue, thick] (F1) -- (F2);
  \draw[blue, thick] (F1) -- (F3);
  \draw[blue, thick] (F1) -- (L7);
  \draw[blue, thick] (F2) -- (L1);
  \draw[blue, thick] (F2) -- (L2);
  \draw[blue, thick] (F3) -- (L3);
  \draw[blue, thick] (F3) -- (L4);
  \draw[blue, thick] (F3) -- (L5);
  \draw[blue, thick] (F3) -- (L6);

  \node at ($(F1)!0.35!(R) +(0.14, 0.04)$)  {\small $e_{-1}$};
  \node at ($(F1)!0.40!(F2)+(0, 0.10)$)     {\small $e_0$};
  \node at ($(F1)!0.40!(F3)+(0.08, 0.08)$)  {\small $e_1$};
  \node at ($(F1)!0.40!(L7)+(0,-0.1)$)     {\small $e_2$};

  \fill[blue] (F1) circle (0.8pt);
  \fill[blue] (F2) circle (0.8pt);
  \fill[blue] (F3) circle (0.8pt);
  \fill[blue] (L1) circle (1pt);
  \fill[blue] (L2) circle (1pt);
  \fill[blue] (L3) circle (1pt);
  \fill[blue] (L4) circle (1pt);
  \fill[blue] (L5) circle (1pt);
  \fill[blue] (L6) circle (1pt);
  \fill[blue] (L7) circle (1pt);
  \fill[red] (R) circle (1.6pt);

  \foreach \i in {0,1,2,3,4,5,6,7} {
    \fill (V\i) circle (1.2pt);
  }
\end{tikzpicture}
\caption{}
\end{subfigure}

    \caption{(a) Construction of the string $s =$``(0)()()'' in $\Upsilon_{3,1}$ from a partial triangulation in $\Omega_{4,3}$: $s = (s(t_0)0s(t_1))s(t_2)$ where $s(t_0)=s(t_1)=\emptyset$ and $s(t_2)=$``()()''. (b) Construction of the string $s=$``(()0(00))'' in $\Upsilon_{3,3}$ from a partial triangulation in $\Omega_{6,3}$: $s = (s(t_0)0s(t_1))s(t_2)$ where $s(t_0)=$``()'', $s(t_1)=$``(00)'', and $s(t_2)=\emptyset$.}
    \label{fig:bij}
\end{figure}
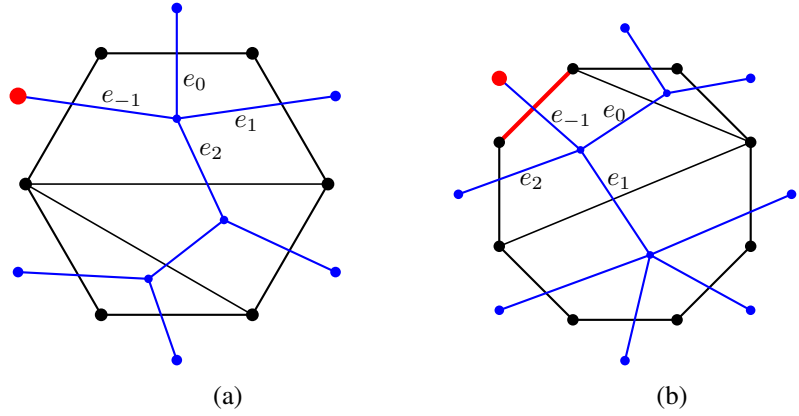

\begin{lemma}\label{lem:bijection}
   For any $1\leq k\leq n$, there exists a bijection between $\Omega_{n,k}$ and $\Upsilon_{k,n-k}$.
\end{lemma}

\begin{proof}
   Root the polygon at a boundary edge.
   The construction of the bijection from $\Omega_{n,k}$ to $\Upsilon_{k,n-k}$ goes via an induction on $k$. 
   For any $n$, if $k=1$ we have the string $(0_10_2\cdots 0_{n-1})$, where we number the zeros only to ease the explanation and to highlight the number of zeros in the string.   
   
   Assume $n\geq 2$ and $k>1$. Let $F$ be the face that contains the root edge, and let $|F|$ denote the number of sides of $F$. Order the edges of $F$ clockwise, so $e_{-1},e_0,\ldots,e_{|F|-2}$ are the edges with $e_{-1}$ being the root edge.
   Let $t_0,t_1,\ldots,t_{|F|-2}$ be the partial triangulations obtained on the other side of each edge (omitting the root edge), and let $t$ be the whole partial triangulation. Then, when $|F|>3$ we construct the string $s(t)$ as 
   $$
      s(t)=(\,s(t_0)\,0_1\,s(t_1)\,0_2\,s(t_2)\,\cdots\,s(t_{|F|-4})\, 0_{|F|-3}\,s(t_{|F|-3})\,)\,s(t_{|F|-2}).
   $$
   If $|F|=3$, we define $s(t)=(s(t_0))s(t_1)$.
   Since each $t_i$ has fewer than $k$ parts, we can apply induction to obtain $s(t_i)$. We set $s(t_i)$ to be the empty string when the other side of $e_i$ is the external face; see Figure~\ref{fig:bij}.
Note that each face $X$ contributes one pair of parentheses and $|X|-3$ zeros. Since the partial triangulation has $k-1$ diagonals, $\sum_X |X|=(n+2)+2(k-1)=n+2k$.
Thus $s(t)$ has $k$ pairs of parentheses and
$\sum_X(|X|-3)=n-k$
zeros, and hence $s(t)\in\Upsilon_{k,n-k}$.

   Now, from each string $s$ we construct a partial triangulation. Note that $s$ must start with '(', which corresponds to the face of the root edge. Then locate the corresponding ')', denote this matching pair $P$, 
   and locate also all the 0's that are surrounded by $P$
   (i.e., that are not surrounded by any other pair of parentheses contained in $P$). In other words, take each pair of open-close parentheses inside $P$ and remove them together with all symbols between them: one will be left with $P$ and a sequence of 
   0's $0_1,0_2,\ldots,0_q$ inside.
   Decompose the string as 
   $$
      s= (\,s_0\,0_1\,s_1\,0_2\,\cdots\,s_{q-1}\,0_q\, s_q\,)\,s_{q+1}.
   $$
   Then the face of the root edge has length $q+3$ and $s_i$ encodes the partial triangulation that lies at the other side of the $(i+1)$-th edge after the root edge in clockwise order.
   Since $s_i$ has fewer pairs of parentheses than $s$, by repeating this procedure we eventually get a string of the form $(0_10_2\cdots 0_q)$; i.e., without any other open-close parentheses inside. This string corresponds to a face of length $q+3$.
 The decomposition above is unique and reverses
the construction of $s(t)$; hence, the two constructions are the inverse of each other.
\end{proof}

With this bijection in hand, we focus on generating a string from $\Upsilon_{k,n-k}$ uniformly at random.
We work with a simple encoding of the strings in $\Upsilon_{k,n-k}$ that keeps the parentheses and then, in any position 
of the string where zeros are allowed, it puts a non-negative integer corresponding to the number of zeros in that position; this will allow us to design an algorithm with running time depending on $k$ instead of $n$.

Given such an encoding of a string in $\Upsilon_{k,n-k}$, one can generate the corresponding partial triangulation, represented by its
$k-1$ diagonals, in $O(k)$ time as follows.
First build a compressed version of the associated rooted tree from the string by scanning the encoding from left to right and interpreting '(' as creating the first child of a vertex, ')' as creating the last child, and each integer as the corresponding number of intermediate children (i.e., zeros in the original string); see Figure~\ref{fig:string-tree-triang}.
Leaves are stored implicitly, by each internal node
storing an integer representing its number of leaf children.
(Note that the preorder depth-first traversal of the tree recovers the string.) 
From the tree, a depth-first search traversal can then be used to generate the partial triangulation 
by adding a boundary edge each time an implicit leaf is reached, in cyclic order starting from the edge corresponding to the ``dummy'' root (which would also be a leaf if the tree were unrooted) and closing a face each time we are done exploring an internal tree vertex. The face is closed so that boundary edges corresponding to all unclosed leaves in the subtree are in the same face, and, except for the root face, the diagonal that closes it is recorded; see Figure~\ref{fig:string-tree-triang}.
Since the compressed tree has $O(k)$ explicit vertices and exactly $k-1$ diagonals are recorded, the traversal takes $O(k)$ time. 

It suffices then to provide 
an algorithm to sample an encoding from $\Upsilon_{k,n-k}$ uniformly at random. We describe our algorithm for this task for $k>1$ next.

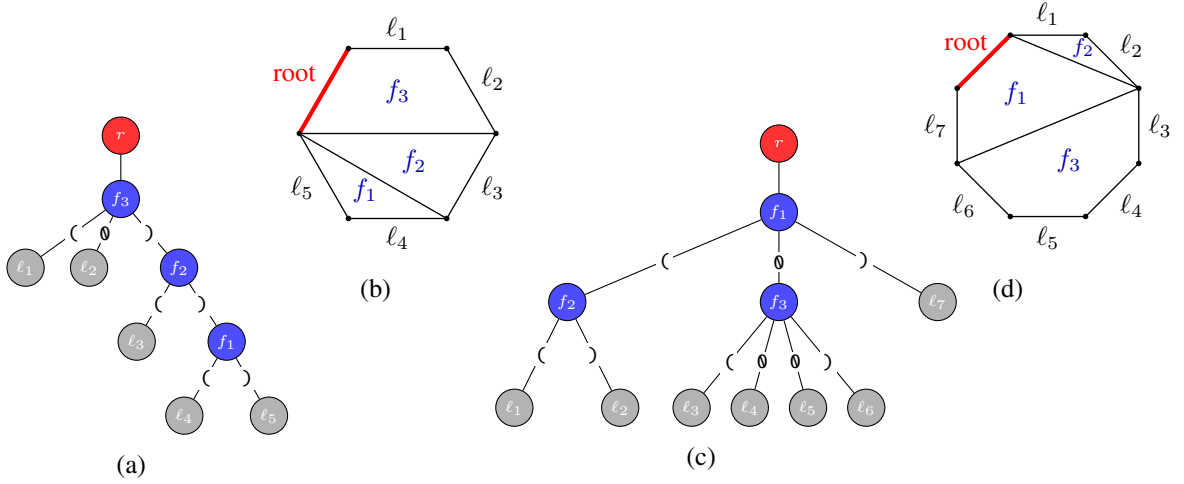
\begin{figure}[t]
  \centering

  \begin{subfigure}[c]{0.20\linewidth}
    \centering
    \begin{tikzpicture}[
      every node/.style={font=\footnotesize},
      internal/.style={circle, draw, fill=blue!70, text=white, inner sep=1pt, minimum size=14pt, font=\footnotesize},
      leaf/.style={circle, draw, fill=gray!60, text=white, inner sep=1pt, minimum size=14pt, font=\scriptsize},
      rootnode/.style={circle, draw, fill=red!80, text=white, inner sep=1pt, minimum size=14pt, font=\footnotesize\bfseries},
      symlbl/.style={font=\scriptsize\ttfamily, inner sep=1pt, fill=white},
      scale=0.7
    ]
      \node[rootnode] (R) at (0, 4.5) {\tiny{$r$}};
      \node[internal] (F3) at (0, 3.3) {\tiny{$f_3$}};
      \node[leaf] (e1) at (-1.8, 2) {\tiny{$\ell_1$}};
      \node[leaf] (e2) at (-0.6, 2) {\tiny{$\ell_2$}};
      \node[internal] (F2) at (1.1, 2) {\tiny{$f_2$}};
      \node[leaf] (e3) at (0.3, 0.6) {\tiny{$\ell_3$}};
      \node[internal] (F1) at (2.0, 0.6) {\tiny{$f_1$}};
      \node[leaf] (e4) at (1.2, -0.8) {\tiny{$\ell_4$}};
      \node[leaf] (e5) at (2.8, -0.8) {\tiny{$\ell_5$}};
      \draw (R) -- (F3);
      \draw (F3) -- node[symlbl] {(} (e1);
      \draw (F3) -- node[symlbl] {0} (e2);
      \draw (F3) -- node[symlbl] {)} (F2);
      \draw (F2) -- node[symlbl] {(} (e3);
      \draw (F2) -- node[symlbl] {)} (F1);
      \draw (F1) -- node[symlbl] {(} (e4);
      \draw (F1) -- node[symlbl] {)} (e5);
    \end{tikzpicture}
    \caption{}
    \label{fig:enc:hex:tree}
  \end{subfigure}
  \begin{subfigure}[b]{0.18\linewidth}
    \centering
    \begin{tikzpicture}[scale=1.3, line width=0.5pt]
      \coordinate (V0) at (240:1);
      \coordinate (V1) at (300:1);
      \coordinate (V2) at (0:1);
      \coordinate (V3) at (60:1);
      \coordinate (V4) at (120:1);
      \coordinate (V5) at (180:1);
   
      \draw (V0) -- (V1);
      \draw (V1) -- (V2);
      \draw (V2) -- (V3);
      \draw (V3) -- (V4);
      \draw[red, ultra thick] (V4) -- (V5);
      \draw (V5) -- (V0);
   
      \draw (V1) -- (V5);
      \draw (V2) -- (V5);
     
      \node[font=\small] at ($(V3)!0.5!(V4) + (0, 0.20)$) {$\ell_1$};
      \node[font=\small] at ($(V2)!0.5!(V3) + (0.22, 0.12)$) {$\ell_2$};
      \node[font=\small] at ($(V1)!0.5!(V2) + (0.22, -0.12)$) {$\ell_3$};
      \node[font=\small] at ($(V5)!0.5!(V0) + (-0.22, -0.12)$) {$\ell_5$};
      \node[font=\small] at ($(V0)!0.5!(V1) + (0, -0.20)$) {$\ell_4$};
   
      \node[red, font=\small] at ($(V4)!0.5!(V5) + (-0.30, 0.18)$) {root};
   
      \node[blue!70!black, font=\small] at (barycentric cs:V0=1,V1=1,V5=1) {$f_1$};
      \node[blue!70!black, font=\small] at (barycentric cs:V1=1,V2=1,V5=1) {$f_2$};
      \node[blue!70!black, font=\small] at (barycentric cs:V2=1,V3=1,V4=1,V5=1) {$f_3$};
      \foreach \v in {V0,V1,V2,V3,V4,V5}{ \fill (\v) circle (0.8pt); }
    \end{tikzpicture}
    \caption{}
    \label{fig:enc:hex:triang}
  \end{subfigure}
  \begin{subfigure}[c]{0.33\linewidth}
    \centering
    \begin{tikzpicture}[
      every node/.style={font=\footnotesize},
      internal/.style={circle, draw, fill=blue!70, text=white, inner sep=1pt, minimum size=14pt, font=\footnotesize},
      leaf/.style={circle, draw, fill=gray!60, text=white, inner sep=1pt, minimum size=14pt, font=\scriptsize},
      rootnode/.style={circle, draw, fill=red!80, text=white, inner sep=1pt, minimum size=14pt, font=\footnotesize\bfseries},
      symlbl/.style={font=\scriptsize\ttfamily, inner sep=1pt, fill=white},
      scale=0.7
    ]
      \node[rootnode] (R) at (0, 6) {\tiny{$r$}};
      \node[internal] (F1) at (0, 4.7) {\tiny{$f_1$}};
      \node[internal] (F2) at (-4.0, 3.0) {\tiny{$f_2$}};
      \node[internal] (F3) at (0, 3.0) {\tiny{$f_3$}};
      \node[leaf] (e7) at (3.0, 3.0) {\tiny{$\ell_7$}};
      \node[leaf] (e1) at (-5.0, 1.0) {\tiny{$\ell_1$}};
      \node[leaf] (e2) at (-3.0, 1.0) {\tiny{$\ell_2$}};
      \node[leaf] (e3) at (-1.65, 1.0) {\tiny{$\ell_3$}};
      \node[leaf] (e4) at (-0.55, 1.0) {\tiny{$\ell_4$}};
      \node[leaf] (e5) at (0.55, 1.0) {\tiny{$\ell_5$}};
      \node[leaf] (e6) at (1.65, 1.0) {\tiny{$\ell_6$}};
      \draw (R) -- (F1);
      \draw (F1) -- node[symlbl, pos=0.55] {(} (F2);
      \draw (F1) -- node[symlbl, pos=0.6] {0} (F3);
      \draw (F1) -- node[symlbl, pos=0.55] {)} (e7);
      \draw (F2) -- node[symlbl] {(} (e1);
      \draw (F2) -- node[symlbl] {)} (e2);
      \draw (F3) -- node[symlbl, pos=0.6] {(} (e3);
      \draw (F3) -- node[symlbl, pos=0.6] {0} (e4);
      \draw (F3) -- node[symlbl, pos=0.6] {0} (e5);
      \draw (F3) -- node[symlbl, pos=0.6] {)} (e6);
    \end{tikzpicture}
    \caption{}
    \label{fig:enc:oct:tree}
  \end{subfigure}
   \begin{subfigure}[b]{0.15\linewidth}
    \centering
    \begin{tikzpicture}[scale=1.3, line width=0.5pt]
      
      \foreach \i in {0,1,2,3,4,5,6,7} {
        \coordinate (V\i) at ({-67.5 + 45*\i}:1);
      }
     
      \draw (V0) -- (V1);
      \draw (V1) -- (V2);
      \draw (V2) -- (V3);
      \draw (V3) -- (V4);
      \draw[red, ultra thick] (V4) -- (V5); 
      \draw (V5) -- (V6);
      \draw (V6) -- (V7);
      \draw (V7) -- (V0);
      \draw (V4) -- (V2);
      \draw (V6) -- (V2);
      \node[font=\small] at ($(V5)!0.5!(V6) + (-0.22, 0)$) {$\ell_7$};
      \node[font=\small] at ($(V6)!0.5!(V7) + (-0.20, -0.16)$) {$\ell_6$};
      \node[font=\small] at ($(V7)!0.5!(V0) + (0, -0.22)$) {$\ell_5$};
      \node[font=\small] at ($(V0)!0.5!(V1) + (0.20, -0.16)$) {$\ell_4$};
      \node[font=\small] at ($(V1)!0.5!(V2) + (0.22, 0)$) {$\ell_3$};
      \node[font=\small] at ($(V2)!0.5!(V3) + (0.20, 0.16)$) {$\ell_2$};
      \node[font=\small] at ($(V3)!0.5!(V4) + (0, 0.22)$) {$\ell_1$};
      \node[red, font=\small] at ($(V4)!0.5!(V5) + (-0.18, 0.20)$) {root};
     
      \node[blue!70!black, font=\small] at (barycentric cs:V3=1,V4=1,V5=1,V7=1) {$f_1$};
      \node[blue!70!black, font=\scriptsize] at ($(barycentric cs:V4=1,V3=1,V2=1)+(0.05,0.05)$) {$f_2$};
      \node[blue!70!black, font=\small] at ($(barycentric cs:V2=1,V3=1,V4=1,V5=1,V6=1)+(0.4,-0.8)$) {$f_3$};
      \foreach \v in {V0,V1,V2,V3,V4,V5,V6,V7}{ \fill (\v) circle (0.8pt); }
    \end{tikzpicture}
    \caption{}
    \label{fig:enc:oct:triang}
  \end{subfigure}

  \caption{Two examples illustrating the construction of a partial triangulation 
  from a string: (a)-(b) correspond to the string ``(0)()()''; (c)-(d) to the string ``(()0(00))''. The root edge of the polygon is marked red, and each tree edge is labeled with its corresponding string symbol. Each leaf $\ell_i$ in the tree 
  corresponds to a boundary edge of the polygon, in the order encountered by a DFS 
  traversal; leaves contribute boundary edges and internal vertices contribute faces when 
  their subtrees are fully explored.}
  \label{fig:string-tree-triang}
\end{figure}

\medskip
\RestyleAlgo{ruled}
\begin{algorithm}[H]
\caption{Uniform sampler for encodings of strings in $\Upsilon_{k,n-k}$}
\label{alg:fixed-k-sampler}
\medskip

\begin{minipage}{0.93\linewidth}
\begin{enumerate}[Step 1.]
\setlength{\itemsep}{0pt}
    \item Sample a balanced parenthesization $P$ with $k$ pairs of open-close parentheses uniformly at random by first
    sampling a full binary tree with $k$ internal nodes uniformly at random using R{\'e}my's algorithm~\cite{remy1985procede} in $O(k)$ time, and then generating the corresponding balanced parenthesization with $k$ pairs of open-close parentheses 
    via the bijection between
    these Catalan structures; see, e.g.,~\cite{Stanley_2015}. 
    
    \item Scan $P$ and compute $r(P)$, the number of admissible positions where zeros may be inserted; note that $r(P) \le 2k-1$.
    
    \item Set
$w(P)=\binom{n-k+r(P)-1}{r(P)-1}$ and $M=\binom{n+k-2}{2k-2}$.

    \item Accept $P$ with probability $\frac{w(P)}{M}$. If $P$ is rejected, return to Step $1$.
    
    \item Once a parenthesization $P$ is accepted, sample a weak composition
$
    (c_1,\ldots,c_{r(P)})
$
of $n-k$ into $r(P)$ non-negative parts uniformly at random.
    
    \item Output the encoding obtained by inserting the counts
$c_1,\ldots,c_{r(P)}$ into the admissible positions of $P$.
\end{enumerate}
\end{minipage}

\end{algorithm}
\medskip

We proceed to analyze the correctness and running time of this algorithm.

\begin{lemma}
\label{lemma:uniform:encoding}
Algorithm~\ref{alg:fixed-k-sampler} outputs the encoding of a uniform random string from
$\Upsilon_{k,n-k}$.
\end{lemma}

\begin{proof}
Fix a balanced parenthesization $P$ with $k$ pairs of open-close parentheses, and recall that we use $r(P)$ to denote the number of admissible positions in $P$ where zeros can be inserted. The number of valid zero placements for $P$ is
\[
    w(P)=\binom{n-k+r(P)-1}{r(P)-1}.
\]
In a uniform sample from $\Upsilon_{k,n-k}$, the parenthesization $P$ must appear with
probability proportional to $w(P)$.  
In Algorithm~\ref{alg:fixed-k-sampler}, the probability that $P$ is generated and accepted in a given trial is
$\frac{1}{C_k}\cdot \frac{w(P)}{M}$. 
Summing over all balanced parenthesizations with $k$ pairs of open-close parentheses, the probability that a trial is
accepted is
\[
    \sum_{P'}
    \frac{1}{C_k}\cdot \frac{w(P')}{M}
    =
    \frac{1}{C_kM}\sum_{P'} w(P').
\]
Therefore, 
\[
    \Pr[P\mid \text{current trial accepted}]
    =
    \frac{
        \frac{1}{C_k}\cdot \frac{w(P)}{M}
    }{
        \frac{1}{C_kM}\sum_{P'}w(P')
    }
    =
    \frac{w(P)}{\sum_{P'}w(P')}.
\]
After a parenthesization $P$ is accepted, Step $5$ chooses one of the $w(P)$ valid zero placements uniformly at random. Thus, conditional on the accepted parenthesization being $P$, each encoding extending $P$ is chosen with probability $1/w(P)$.
Moreover, for an encoding $\sigma$ of a string in $\Upsilon_{k,n-k}$, 
let $P(\sigma)$ be its underlying balanced parenthesization. This parenthesization is uniquely determined by $\sigma$, by deleting the zeros and keeping only the parentheses. Therefore, for the algorithm to output $\sigma$, it must first accept $P(\sigma)$ and then choose in Step $5$ the zero placement that produces $\sigma$. Thus $\sigma$ is the output of the algorithm with probability:
\[
    \frac{w(P(\sigma))}{\sum_{P'}w(P')}
    \cdot
    \frac{1}{w(P(\sigma))}
    =
    \frac{1}{\sum_{P'}w(P')} = \frac{1}{|\Upsilon_{k,n-k}|},
\]
as claimed.
\end{proof}

\begin{lemma}
\label{lemma:alg2:rt}
Algorithm~\ref{alg:fixed-k-sampler} has expected running time $O(k\log k)$.
\end{lemma}

\begin{proof}
R{\'e}my's algorithm runs in $O(k)$ time, so Step 1 takes $O(k)$ time. Likewise, Step 2 does a single scan, and so it also takes $O(k)$ time.
The acceptance probability $\frac{w(P)}{M}$ can  be computed in $O(k)$ time as well since letting $r = r(P)$ we have
\[
    \frac{w(P)}{M}
    =
    \frac{\binom{n-k+r-1}{r-1}}
         {\binom{n+k-2}{2k-2}}
    =
    \frac{\prod_{j=0}^{2k-2-r} (r+j)}{\prod_{j=0}^{2k-r-2} (n-k+r+j)},
\]
and then both the denominator and numerator can be computed with $O(k)$ arithmetic operations.

We show next that the acceptance probability is at least constant.   
Since for $r\ge 1$
\[
    \frac{\binom{n-k+r}{r}}{\binom{n-k+r-1}{r-1}}
    =
    \frac{n-k+r}{r}
    \ge 1,
\]
the binomial coefficient $\binom{n-k+r-1}{r-1}$
is increasing in $r$. Therefore,
\[
    w(P)\le M=\binom{n+k-2}{2k-2}.
\]
If $P$ has the form $(\ldots)$, 
then $r =  2k-1$ and $w(P)=\binom{n+k-2}{2k-2}=M$.
Therefore, every such balanced parenthesization is accepted with probability exactly $1$.
The number of $(\ldots)$ balanced parenthesizations is $C_{k-1}$, so the probability that a balanced parenthesization sampled uniformly at random has the
form $(\ldots)$ is 
\[
    \frac{C_{k-1}}{C_k}
    =
    \frac{k+1}{2(2k-1)} \ge \frac{1}{4}.
\]
Since every such balanced parenthesization is accepted, the acceptance probability in Step 4 is at
least $1/4$ and thus the expected number of trials is at most $4$.

For Step 5, 
after accepting a parenthesization $P$,
we sample a uniform random composition of $n-k$ elements (zeros) into $r$ non-negative parts. This can be done by selecting a subset of size $r-1$ uniformly at random among all subsets of that size in $\{1,\dots,n-k+r-1\}$. For this, we can use Floyd's algorithm~\cite{bentley1987programming} 
which produces a random subset 
$S\subseteq\{1,\ldots,n-k+r-1\}$ of size $r-1$ in $O(r \log r)$ time.
Sorting the elements of $S$ gives
$s_1<\cdots<s_{r-1}$ in $O(r\log r)$ time, and 
we set
$c_1=s_1-1$, $c_i=s_i-s_{i-1}-1$ for $2\le i\le r-1$, and 
$c_r=n-k+r-1-s_{r-1}$.
All combined, the expected running time of the algorithm is $O(k \log k)$.
\end{proof}

We conclude this section with the proof of our main result Theorem~\ref{thm:intro}.

\begin{proof}[Proof of Theorem~\ref{thm:intro}]
Our algorithm first runs Algorithm~\ref{alg:exact:k:sampler} to sample the number $X$ of diagonals, then runs Algorithm~\ref{alg:fixed-k-sampler} with $k=X+1$, and finally converts the resulting encoding into the corresponding partial triangulation using Lemma~\ref{lem:bijection}.  
The correctness of the algorithm follows from Lemmas~\ref{lemma:alg:exact:k:sampler:correctness}, \ref{lem:bijection}, and~\ref{lemma:uniform:encoding}.
The running time bound follows from Lemmas~\ref{lemma:alg1:rt}, \ref{lemma:expectation:bound}, and ~\ref{lemma:alg2:rt}.
\end{proof}

\section*{Acknowledgments}

This work was started at the 2021 AIM SQuaRE workshop ``Connections between computational and physical phase transitions.'' We thank Sarah Cannon, Tyler Helmuth, and Will Perkins for many helpful discussions. We also thank an anonymous referee whose comments lead to a much simpler proof of Theorem~\ref{thm:intro}.

\section*{AI Disclosure}

All proofs in the paper were obtained without the use of generative AI. The writing was also done entirely without the use of large language models.

\bibliographystyle{alpha}
\bibliography{references}

@article{cayley1890partitions,
  title={On the partitions of a polygon},
  author={A. Cayley},
  journal={Proceedings of the London Mathematical Society},
  volume={1},
  number={1},
  pages={237--264},
  year={1890},
  publisher={Oxford University Press}
}

@inbook{Stanley_2015, 
place={Cambridge}, 
title={Catalan Numbers}, 
publisher={Cambridge University Press}, 
author={Stanley, Richard P.}, 
year={2015}, 
pages={201–204}}

@article{remy1985procede,
  title={Un proc{\'e}d{\'e} it{\'e}ratif de d{\'e}nombrement d'arbres binaires et son application {\`a} leur g{\'e}n{\'e}ration al{\'e}atoire},
  author={R{\'e}my, Jean-Luc},
  journal={RAIRO. Informatique th{\'e}orique},
  volume={19},
  number={2},
  pages={179--195},
  year={1985},
  publisher={EDP Sciences}
}

@inproceedings{cannon2014phase,
  title={Phase Transitions in Random Dyadic Tilings and Rectangular Dissections},
  author={Cannon, Sarah and Miracle, Sarah and Randall, Dana},
  booktitle={Proceedings of the Twenty-Sixth Annual ACM-SIAM Symposium on Discrete Algorithms (SODA)},
  pages={1573--1589},
  year={2014},
  organization={SIAM}
}

@inproceedings{cannon2017polynomial,
  title={Polynomial Mixing of the Edge-Flip Markov Chain for Unbiased Dyadic Tilings},
  author={Cannon, Sarah and Levin, David A and Stauffer, Alexandre},
  booktitle={Proceedings of RANDOM},
  year={2017}
}

@article{angel2014phase,
  title={The phase transition for dyadic tilings},
  author={Angel, Omer and Holroyd, Alexander and Kozma, Gady and W{\"a}stlund, Johan and Winkler, Peter},
  journal={Transactions of the American Mathematical Society},
  volume={366},
  number={2},
  pages={1029--1046},
  year={2014}
}

@article{lagarias2002counting,
  title={Counting dyadic equipartitions of the unit square},
  author={Lagarias, Jeffrey C and Spencer, Joel H and Vinson, Jade P},
  journal={Discrete mathematics},
  volume={257},
  number={2-3},
  pages={481--499},
  year={2002},
  publisher={Elsevier}
}

@inproceedings{caputo2013random,
  title={Random lattice triangulations: structure and algorithms},
  author={Caputo, Pietro and Martinelli, Fabio and Sinclair, Alistair and Stauffer, Alexandre},
  booktitle={Proceedings of the Forty-fifth annual ACM Symposium on Theory of Computing (STOC)},
  pages={615--624},
  year={2013}
}

@article{stauffer2017lyapunov,
  title={A {Lyapunov} function for Glauber dynamics on lattice triangulations},
  author={Stauffer, Alexandre},
  journal={Probability Theory and Related Fields},
  volume={169},
  number={1},
  pages={469--521},
  year={2017},
  publisher={Springer}
}

@article{10.1214/16-EJP4321,
author = {Pietro Caputo and Fabio Martinelli and Alistair Sinclair and Alexandre Stauffer},
title = {{Dynamics of lattice triangulations on thin rectangles}},
volume = {21},
journal = {Electronic Journal of Probability},
number = {none},
publisher = {Institute of Mathematical Statistics and Bernoulli Society},
pages = {1 -- 22},
year = {2016},
doi = {10.1214/16-EJP4321},
URL = {https://doi.org/10.1214/16-EJP4321}
}

@book{de2010triangulations,
  title={Triangulations: structures for algorithms and applications},
  author={De Loera, Jes{\'u}s and Rambau, J{\"o}rg and Santos, Francisco},
  year={2010},
  publisher={Springer Science \& Business Media}
}

@article{ya1989real,
  title={Real plane algebraic curves: Constructions with controlled topology},
  author={Ya, Viro O},
  journal={Algebra i Analiz},
  volume={1},
  pages={1--73},
  year={1989}
}

@incollection{dais2002resolving,
  author    = {Dais, D. I.},
  title     = {Resolving 3-Dimensional Toric Singularities},
  booktitle = {Geometry of Toric Varieties},
  series    = {S{\'e}minaires et Congr{\`e}s},
  volume    = {6},
  pages     = {155--186},
  publisher = {Soci{\'e}t{\'e} Math{\'e}matique de France},
  address   = {Paris},
  year      = {2002},
  mrnumber  = {2075609}
}

@incollection{gelfand1994discriminants,
  title={A-discriminants},
  author={Gelfand, Israel M and Kapranov, Mikhail M and Zelevinsky, Andrei V},
  booktitle={Discriminants, resultants, and multidimensional determinants},
  pages={271--296},
  year={1994},
  publisher={Springer}
}

@article{mcshine1997mixing,
  title={On the mixing time of the triangulation walk and other Catalan structures.},
  author={McShine, Lisa and Tetali, Prasad},
  journal={Randomization methods in algorithm design},
  volume={43},
  pages={147--160},
  year={1997}
}

@inproceedings{eppstein2023improved,
  title={Improved Mixing for the Convex Polygon Triangulation Flip Walk},
  author={Eppstein, David and Frishberg, Daniel},
  booktitle={50th International Colloquium on Automata, Languages, and Programming (ICALP)},
  pages={56--1},
  year={2023}
}

@article{molloy1999mixing,
  title={On the mixing rate of the triangulation walk},
  author={Molloy, Michael and Reed, Bruce and Steiger, William},
  journal={DIMACS Series in Disc. Math. and Theoret. Comput. Sci},
  volume={43},
  pages={179--190},
  year={1999}
}

@inproceedings{alev2026faster,
  author={Alev, Vedat Levi and Frishberg, Daniel and Sarantis, Mihalis and Tetali, Prasad},
  title     = {Faster Mixing for Triangulations via Transport Flows},
  booktitle = {53rd International Colloquium on Automata, Languages, and Programming (ICALP)},
  year      = {2026}
}

@article{bentley1987programming,
  title={Programming pearls: a sample of brilliance},
  author={Bentley, Jon and Floyd, Bob},
  journal={Communications of the ACM},
  volume={30},
  number={9},
  pages={754--757},
  year={1987},
  publisher={ACM New York, NY, USA}
}

@book{lee2017subdivisions,
  title={Handbook of discrete and computational geometry},
  author={Toth, Csaba D and O'Rourke, Joseph and Goodman, Jacob E},
  year={2017},
  publisher={CRC press}
}

@article{lee1989associahedron,
  title={The associahedron and triangulations of the n-gon},
  author={Lee, Carl W},
  journal={European Journal of Combinatorics},
  volume={10},
  number={6},
  pages={551--560},
  year={1989},
  publisher={Elsevier}
}

@article{kirkman1857k,
  title={On the k-partitions of the r-gon and r-ace},
  author={Kirkman, Thomas P},
  journal={Phil. Trans. Royal Soc. London},
  volume={147},
  pages={217--272},
  year={1857}
}

@article{stanley1996polygon,
  title={Polygon dissections and standard Young tableaux},
  author={Stanley, Richard P},
  journal={journal of combinatorial theory, Series A},
  volume={76},
  number={1},
  pages={175--177},
  year={1996},
  publisher={Elsevier}
}

@article{przytycki1998polygon,
  title={Polygon Dissections and Euler, Fuss, Kirkman, and Cayley Numbers},
  author={Przytycki, J{\'o}zef H and Sikora, Adam S},
  journal={Journal of Combinatorial Theory, Series A},
  volume={92},
  number={1},
  pages={68--76},
  year={2000},
  publisher={Elsevier}
}

@article{aldous1994triangulating,
  title={Triangulating the circle, at random},
  author={Aldous, David},
  journal={The American Mathematical Monthly},
  volume={101},
  number={3},
  pages={223--233},
  year={1994},
  publisher={Taylor \& Francis}
}

@article{panagiotou2023exact,
  title={Exact-size sampling of enriched trees in linear time},
  author={Panagiotou, Konstantinos and Ramzews, Leon and Stufler, Benedikt},
  journal={SIAM Journal on Computing},
  volume={52},
  number={5},
  pages={1097--1131},
  year={2023},
  publisher={SIAM}
}

@article{Stufler,
author = {Benedikt Stufler},
title = {{Limits of random tree-like discrete structures}},
volume = {17},
journal = {Probability Surveys},
number = {none},
publisher = {Institute of Mathematical Statistics and Bernoulli Society},
pages = {318 -- 477},
year = {2020},
doi = {10.1214/19-PS338},
URL = {https://doi.org/10.1214/19-PS338}
}

\end{document}